\documentclass[times]{oupau}

\usepackage{booktabs} 

\usepackage{lipsum}
\usepackage{graphicx}

\usepackage{algorithm}
\usepackage{algorithmic}

\usepackage{amssymb}
\usepackage{amsmath}
\usepackage{graphicx}

\usepackage{stmaryrd}

\usepackage{mathrsfs}
\usepackage{enumitem}
\usepackage{multicol}
\usepackage{longtable}
\usepackage{breqn}
\usepackage{hyperref}
\usepackage{appendix}

\newcommand{\cR}{\mathcal{R}}

\newcommand{\K}{\mathbb{K}}
\newcommand{\N}{\mathbb{N}}

\begin{document}

\title{On the Algorithmic Verification of Nonlinear Superposition for Systems of First Order Ordinary Differential Equations}
\shorttitle{Algorithmic Verification of Nonlinear Superposition}

\volumeyear{}
\paperID{}

\author{Veronika~Treumova\affil{1}, Dmitry~A.~Lyakhov\affil{2}, and Dominik~L.~Michels\affil{2}}
\abbrevauthor{V. Treumova, D.~A.~Lyakhov, and D.~L.~Michels}
\headabbrevauthor{V. Treumova et al.}

\address{%
\affilnum{1}National Institute for Pure and Applied Mathematics, Rio de Janeiro, Brazil\\
\affilnum{2}Computational Sciences Group, KAUST, KSA}

\correspdetails{\affil{1}veronika.treumova@impa.br and \affil{2}\{dmitry.lyakhov, dominik.michels\,\}@kaust.edu.sa}

\received{}
\revised{}
\accepted{}

\communicated{}

\maketitle

\section*{Abstract}

This paper belongs to a group of work in the intersection of symbolic computation and group analysis aiming for the symbolic analysis of differential equations. The goal is to extract important properties without finding the explicit general solution. In this contribution, we introduce the algorithmic verification of nonlinear superposition properties and its implementation. More exactly, for a system of nonlinear ordinary differential equations of first order with a polynomial right-hand side, we check if the differential system admits a general solution by means of a superposition rule and a certain number of particular solutions. It is based on the theory of Newton polytopes and associated symbolic computation. The developed method provides the basis for the identification of nonlinear superpositions within a given system and for the construction of numerical methods which preserve important algebraic properties at the numerical level.\\\\
{\bf Keywords:} Algorithmic Verification, Difference Thomas Decomposition, Newton Polytope, Nonlinear Superposition, Strong Consistency, Vessiot-Guldberg-Lie Algebra.

\section{Introduction}

Nonlinear ordinary differential equations (ODEs) play a crucial role in modeling a wide range of complex phenomena across various scientific disciplines. Unlike linear ODEs, which are well-understood and often analytically solvable in terms of elementary functions, nonlinear ODEs capture intricate interactions and dependencies that characterize many real-world systems ranging from physics and engineering to biology and economics. Generally, addressing these equations involves employing difference approximations for differential systems since analytical methods face significant constraints. Rather than seeking exact solutions, engineering tasks require obtaining an approximate solution with a specified level of precision. Furthermore, practical considerations often necessitate the use of an appropriate method that not only maintains the properties of the underlying system in the limit but also does so for finite time steps, achieving an approximation.

In this contribution, we focus on the symbolic analysis of nonlinear superposition properties of systems of ODEs. Nonlinear superposition is a generalization of the superposition principle to nonlinear systems. It states that, for certain nonlinear systems, the response to a combination of inputs can be expressed as a combination of the responses to the individual inputs. This means that the system's behavior is not always linear but still exhibits some degree of predictability. The nonlinear superposition property is important in many areas of physics, including quantum mechanics, fluid dynamics, and optics \cite{li2023nonlinear,bassi2013models,ackermann2006nonlinear,rogers2018madelung}. It is used to study the behavior of complex systems that cannot be easily described by linear equations. For instance, in quantum mechanics, the nonlinear superposition property is used to explain the behavior of particles that can exist in multiple states at the same time. This key element of modern quantum mechanics has many important consequences, such as the existence of quantum entanglement and wave-particle duality. In fluid dynamics, the nonlinear superposition property is used to study the behavior of waves in fluids. For example, it can be used to explain the formation of solitons, 
the interesting phenomenon of a single nonlinear wave that can travel long distances without changing shape. In optics, the nonlinear superposition property is used to explain the behavior of light in nonlinear media -- the phenomenon of self-focusing, i.e., a beam of light can focus itself as it travels through a nonlinear medium.

This paper is organized as follows. In Section~\ref{sec:MathFormulation}, we provide all required definitions and the basic theorem. After discussing related work in Section~\ref{sec:RelatedWork}, we briefly present an elementary example illustrating the relevance of the nonlinear superposition property in Section~\ref{sec:Riccati}. The main mathematical objects we deal with and the algorithms are then presented in Section~\ref{sec:Alg}. The implementation of these algorithms in Maple is described in Section~\ref{sec:Implement}, and a nontrivial example is illustrated in Section~\ref{sec:Nontrivial}. Finally, we provide our view on applications in Sections~\ref{sec:Thomas} and \ref{sec:MatrixRicatti}, and generalizations in Section~\ref{sec:Generalization} before our concluding remarks in Section~\ref{sec:Conclusion}.

\section{Mathematical Formulation}
\label{sec:MathFormulation}

Sophus Lie had an extraordinary geometric vision of problems. It helped him to simplify analytical calculations and often led him to
new theoretical concepts. One of them is the generalization of properties of linear equations, which led to the concept of nonlinear superposition. 
\begin{definition}
We say that a differential system admits a nonlinear superposition rule if its general solution can be represented in the form
$$
\textbf{x} = \Phi(\textbf{x}^1, \textbf{x}^2, \dots, \textbf{x}^m; C_1, \dots, C_n)\,,
$$
including a finite number $m$ of partial solutions and $n$ arbitrary constants.
\end{definition}

The nonlinear superposition problem is strongly connected to the field of Lie algebras \cite{lie1885allgemeine}. Together with Vessiot and Guldberg, Sophus Lie characterized all possible first-order ODE systems with superposition property:
\begin{theorem}
The system of first-order ODEs with a vector of dependent variables $\textbf{x} = (x_1, \dots, x_n)$ admits a nonlinear superposition if and only if it has the form of generalized separation of variables with $r < \infty$ members:
\begin{equation}\label{generalized-separation-variables-equation}
\frac{d \textbf{x}}{dt} = T_1(t) \xi^{i}_{1}(\textbf{x}) + T_2(t) \xi^{i}_{2}(\textbf{x}) + \dots + T_r(t) \xi^{i}_{r}(\textbf{x})\,,
\end{equation}
where the operators\footnote{We use the common summation rule in repeated indices.}
$$
X_{\alpha} = \xi_{\alpha}^i(\textbf{x}) \frac{\partial}{\partial x_i}, \,\alpha = 1, \dots r\,,
$$
satisfy the commutator relations
\begin{equation}\label{VGL}
[ X_{\alpha}, X_{\beta} ] = C_{\alpha \beta}^{\gamma} X_{\gamma}\,,
\end{equation}
in $C_{\alpha \beta}^{\gamma}$, which are called the structure constants. The equations mean that the operators span a Lie algebra $L_r$ of the finite dimension $r$. The number $m$ of necessary particular solutions is estimated by $nm \geq r$.
\end{theorem}

The linear space (\ref{VGL}) is called Vessiot-Guldberg-Lie algebra to honor the memory of pioneers. Significant efforts were invested in classifying all systems with nonlinear superposition property, but up to now, only realizations of small dimensions like $3$ and $4$ have been obtained \cite{ibragimov2016three, ibragimov2017classification}.

However, the actual point of how to check the conditions of the theorem was missing. In this paper, we will derive an algorithm for verifying the condition of closeness for the important case of polynomially nonlinear $\xi^{i}_{k}(\textbf{x})$ functions. It is based on the theory of Newton polytopes and represents a constructive procedure to verify -- by a finite amount of steps -- if the Lie algebra is finite or not. Such a verification requires a huge amount of symbolic algebraic computations and is only possible to do using a computer algebra system. The general logic here is similar to a previous ISSAC paper \cite{lyakhov2017algorithmic} where linearizability criteria were derived. However, we use a completely different algebraic technique.

\begin{figure}
\centering
\includegraphics[width=0.45\textwidth]{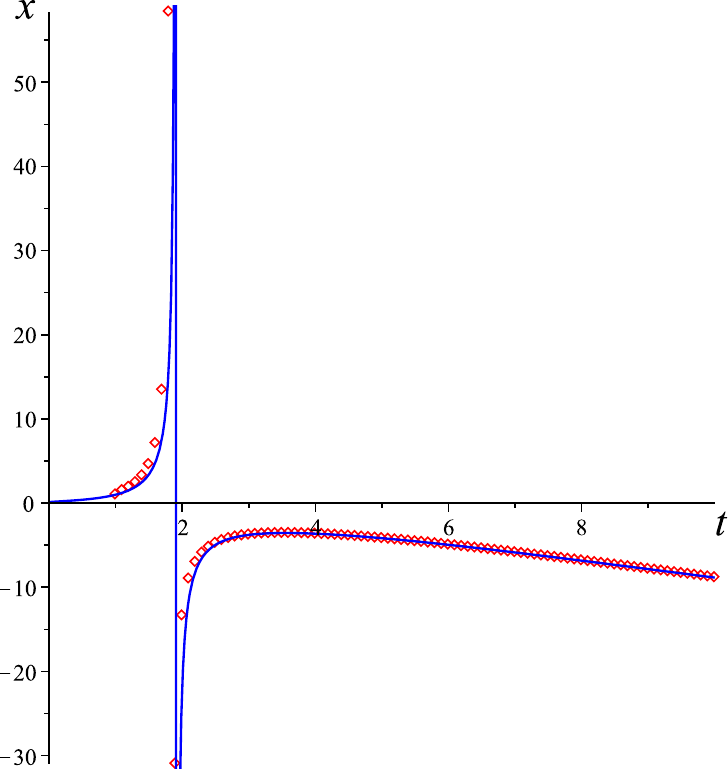}
\caption{Illustration of the exact (blue curve) and the numerical (red rhombi) solution of the Ricatti equation with coefficients (\ref{coeffs}).}
\label{fig:Ricatti}
\end{figure}

\section{Related Work}
\label{sec:RelatedWork}

There is a burgeoning interest in symbolic computation in disciplines such as bioinformatics \cite{fajiculay2022biosans} and biochemistry \cite{boulier2019symbiont}. Preprocessing of the differential systems is beneficial from a numerical perspective as recent work has shown the transformation of polynomial dynamical systems into quadratic ones \cite{bychkov2021optimal}, and exact hierarchical reductions to smaller and more tractable blocks by linear transformations \cite{demin2024exact}. Classical computer algebra achievements such as the Rosenfeld-Gr\"obner algorithm \cite{boulier1995representation} and the Differential Thomas Decomposition \cite{thomas1937differential,bachler2012algorithmic} are ubiquitously used now for the symbolic analysis of differential equations. It started from basic tasks in classical mathematical physics for partial differential systems like constructing all compatibility conditions, detecting the inconsistency of the system, deriving degrees of freedom in the general solution, and determining the coefficients of a power series solution. The general work in this direction is within the scope of Hilbert's Program for the algebraization of mathematics. 

The work of Sophus Lie on group analysis of differential equations was continued in the 20th and 21st centuries by Lev Ovsyannikov and his multiple followers, who applied Lie's method for the large-scale integration of differential equations that arose in practice. There has been recent interest in theoretical studies of the nonlinear superposition principle. E.g., Gainetdinova et al.~generalized the superposition theorem for systems of two second-order ODEs admitting four-dimensional Lie algebras \cite{gainetdinova2017integrability}. Carinena et al.~investigated an important class of constrained systems on Dirac manifolds, more precisely called Dirac-Lie systems, which also preserve the Hamiltonian structure and superposition rule \cite{carinena2014dirac}. Moreover, Ballesteros et al. \cite{ballesteros2021poisson} studied Lie-Hamiltonian systems endowed with the Poisson–Hopf deformation formalism, which extends nonlinear superposition theory to quantum mathematical physics.

\section{Riccati Equation}
\label{sec:Riccati}
The Riccati equation can be considered as the simplest example which admits a nonlinear superposition property
\begin{equation}\label{scalar-ricatti}
\frac{dx}{dt} = a_0(t) + a_1(t) x + a_2(t) x^2\,,
\end{equation}
so that any solution of it could be represented as a rational function of three known solutions $x_i(t), i = 1,\dots,4$ and a single constant $C$:
$$
x(t) = \frac{C (\mathit{x_1} \mathit{x_3} - \mathit{x_2} \mathit{x_3}) +\mathit{x_1} \mathit{x_3} -\mathit{x_1} \mathit{x_2} }{C (\mathit{x_1} - \mathit{x_2}) +\mathit{x_3} - \mathit{x_2}}\,.
$$
The superposition rule is equivalent to the symmetric property that the subharmonic ratio of four solutions remains the same for all $t$:
\begin{equation}\label{subharmonic-ratio}
\frac{x_4 - x_1}{x_1 - x_2} : \frac{x_4 - x_3}{x_2 - x_3} = C\,.
\end{equation}
The simple forward Euler scheme 
$$
\frac{x_{n+1} - x_{n}}{h} = a_0(t_n) + a_1(t_n) x_n + a_2(t_n) x_n^2\,,
$$
where $t_{n+1} = t_n + h$, according to a general logic of numerical methods, converges on intervals without singularities. However, it does not satisfy property (\ref{subharmonic-ratio}), and when approaching the infinity point $x(t) \rightarrow \infty$, the solution crashes.

In contrast, the slightly modified version -- the partially implicit Euler scheme
\begin{equation}\label{good-scheme}
\frac{x_{n+1} - x_{n}}{h} = a_0(t_n) + a_1(t_n) x_n + a_2(t_n) x_n x_{n+1}
\end{equation}
satisfies property (\ref{subharmonic-ratio}), and the solution can cross singularity points without an excessive error accumulation.

For example, let us consider in (\ref{scalar-ricatti}),
\begin{equation}\label{coeffs}
a_2(t) = \frac{1}{t},\,\,\,\,a_1(t) = 1,\,\,\,\,a_0(t) = 0\,.
\end{equation}
The general solution is given by means of exponential integrals\footnote{The function is not elementary as shown by applying the Risch algorithm \cite{risch1969problem}.}
$$
x(t) = (Ei(1, -t) + C)^{-1} {e^t}\,,
$$
where
$$Ei(a, z) = \int_1^{\infty} e^{-\tau z} \tau^{-a} d \tau\,.
$$
The solution of the Cauchy problem (\ref{scalar-ricatti}) with initial condition $x(1) =~1$ is illustrated in Figure~\ref{fig:Ricatti} in which the blue curve corresponds to the exact solution and the red rhombi represent the numerical solution obtained from the finite-difference scheme (\ref{good-scheme}).

In the next section, we consider the important class of nonlinear dynamical systems of the form (\ref{generalized-separation-variables-equation}) where
$$\xi^{i}_{k}(x) = \sum_{m=0}^N a_m(i,k) x^m$$
are multivariate polynomials of dependent variables. In this case, we can say if the corresponding set of differential operators forms a Lie algebra or if it does not admit finite closure by means of a Lie bracket.

\section{Algorithmics}
\label{sec:Alg}

This section introduces the primary mathematical entities we work with, along with the corresponding algorithms. The main object of consideration here are polynomial vector fields and their associated first-order differential operators:
$$
X = f_1 \frac{\partial}{\partial x_1} + \dots + f_n \frac{\partial}{\partial x_n}\,,
$$
where $f_i = f_i(x_1, x_2, \dots, x_n)$ -- polynomial function of all variables. The goal here is to understand if a finite set of operators
$$
\{X_1, X_2, \dots, X_m\}
$$
is closed under the Lie bracket or not. We will start from the scalar case when the dependent variable is only a single one.

\subsection{One-dimensional Case}
The one-dimensional case is not difficult to investigate, but at the same time, it provides an idea for the general case. 

Let $X_1,\dots,X_r,$ be operators, where $r>1.$ For each $i$, we denote $X_i=~f_i\frac{\partial}{\partial x}.$ We define the degree of operators as $deg(X_i):=deg(f_i)$.

The algorithm for checking nonlinear superposition is given by the following steps.
\begin{enumerate}
    \item Check degrees of the polynomials $\{f_i\}.$ If there are two polynomials $g$ and $h$ of different degrees which are greater or equal to 2, then algebra is infinite-dimensional. Otherwise, go to (2).
    \item Compute all pairwise commutators. Check the dimensions in the vector space spanned by $\{x^k\}_{k=0}^\infty$:
    \begin{itemize}
        \item If $dim(\langle f_i,[f_j,f_k]\rangle)=dim(\langle f_i\rangle),$ then the algebra is finite-dimensional.
        \item Otherwise go to (1) with the set $\{f_i,[f_j,f_k]\}$.
    \end{itemize} 
\end{enumerate}
The algorithm always stops working due to the criterion.

\begin{theorem}
    The algebra $\mathcal{A}$ spanned by $\{X_i\}$ where $X_i=f_i\frac{\partial}{\partial x}$ is infinite-dimensional if and only if, after adding all pairwise commutators a finite number of times, there will be two polynomials $g$ and $h$ of different degrees which are greater or equal to 2.
\end{theorem}

\begin{proof} Firstly, observe that $\mathcal{A}$ is infinite-dimensional if and only if $\forall\; C>0\;\exists\;X=f\frac{\partial}{\partial x}\in\mathcal{A}$ such that $deg(f)>C.$
Therefore, sufficiency of the condition follows from the fact that the commutator of operators of degrees $d_1, d_2$ where $d_1\neq d_2,$ has degree $d_1+d_2-1.$ The condition is necessary as well since in the infinite-dimensional algebra, there always exist two operators of different degrees which are greater or equal to 2.
\end{proof}

\subsection{General Case}

The general case is treated in the same manner. The analog of polynomials from the previous section is Newton polytope.

We start with convenient notations. To a vector $(y_1,\dots,y_d)$ corresponds a linear vector field $Y = y_1 D_1 + \dots + y_d D_d,$ where $D_k = x_k\frac{\partial}{\partial x_k}.$ 

Any operator $X$ can be written as a Laurent polynomial, whose coefficients are linear vector fields $$X=\sum_{n\in \mathbb{Z}^d}x^nX(n)\,,$$
where $x^n=x_1^{n_1}\dots x_d^{n_d},\,X(n)=\sum\limits_{i=1}^dX_{n,i}D_i.$ We call $(X_{n,1},\dots,X_{n,d})$ the vector corresponding to $n.$ Observe that $X(n)\neq0$ only for a finite number of $n\in \mathbb{Z}^d.$ Thus, for a given operator $X$, we can construct a convex polytope $$N_X=Conv\{n\:|\: X(n)\neq0\}\,,$$
which is called a \textit{Newton polytope}.

Let again $X_1,\dots,X_r$ be operators, where $r>1,$ and $N_{X_1},\dots,N_{X_r}$ are the corresponding polytopes. The algorithm for checking nonlinear superposition is given by the following steps.
\begin{enumerate}
\item Check the vertices of polytopes $\{N_{X_k}\}$. If there are two vertices $v\in vert(N_{X_i})$ and $u\in vert(N_{X_j})$, $i\neq j$, with corresponding vectors $V$ and $U$ such that all five conditions
\begin{itemize}
    \item $|v+u|>max(|v|,|u|)$\,,
    \item $v+u$ is the vertex of $N_{X_i}+N_{X_j}$ (Minkowski sum)\,,
    \item $(u\cdot V)U - (v\cdot U)V\neq0$, with dot product $(\cdot)$\,,
    \item if $(u\cdot V)=0$ then $\forall m\in\mathbb{N}_{\geq1}\;(v+mu,U)\neq0$\,,
    \item if $(v\cdot U)=0$ then $\forall m\in\mathbb{N}_{\geq1}\;(mv+u,V)\neq0$\,,
\end{itemize}   
are satisfied, then the algebra is infinite-dimensional.

\item Compute all pairwise commutators and check dimensions:

\begin{itemize}
    \item If $dim(\langle \Xi_{X_\alpha}\rangle)=dim(\langle \Xi_{X_\alpha},\Xi_{[X_\beta,X_\gamma]} \rangle),$ then the algebra is finite-dimensional. 
    \item Otherwise go to (1) with operator set $\{X_\alpha,[X_\beta,X_\gamma]\}$.
\end{itemize}

\end{enumerate}

The algorithm always terminates due to the following criterion.
\begin{theorem}
    Algebra $\mathcal{A}$ spanned by $\{X_\alpha\}$ where $X_\alpha=\sum\limits_{i=1}^d f_{\alpha,i}\frac{\partial}{\partial x_i}$ is infinite-dimensional if and only if, after adding all pairwise commutators a finite number of times, there will be two vertices $v\in vert(N_{X_i}), u\in vert(N_{X_j})$, $i\neq j$, with corresponding vectors $V$ and $U$ such that all following five conditions are satisfied:
\begin{enumerate}[label=(\roman*)]
    \item $|v+u|>max(|v|,|u|)$\,,
    \item $v+u$ is the vertex of $N_{X_i}+N_{X_j}$ (Minkowski sum)\,,
    \item $(u\cdot V)U - (v\cdot U)V\neq0$, with dot product $(\cdot)$\,,
    \item if $(u\cdot V)=0$ then $\forall m\in\mathbb{N}_{\geq1}\;(v+mu,U)\neq0$\,,
    \item if $(v\cdot U)=0$ then $\forall m\in\mathbb{N}_{\geq1}\;(mv+u,V)\neq0$\,.
\end{enumerate}   
\end{theorem}

\begin{proof}
Firstly, observe that $\mathcal{A}$ is infinite-dimensional if and only if $\forall\; C>0\;\exists\;X=\sum\limits_{i=1}^d f_i\frac{\partial}{\partial x_i}\in\mathcal{A}$ such that\\
$$max(\{|v| : v\in vert(N_X)\})>C\,.$$

\textit{Sufficiency.} Suppose $a$ and $b$ are vertices satisfying the conditions (i)-(v). Please note, that $$[x^vV,x^uU] = x^{v+u}(u\cdot V)U - (v\cdot U)V)\,.$$
With $K_{u,v}=(u\cdot V)U - (v\cdot U)V$, it follows that 
\begin{equation*}
\begin{aligned}
    N_{[X,Y]}\supset Conv\{&v+u\:| v\in \,vert(X),\,u\in vert(Y)\,,\\
\: \,&v+u\in vert(X+Y),\,K_{u,v}\neq0\}\,.
\end{aligned}
\end{equation*}
Thus, it is easy to deduce that either the set $\{|a+mb|\}_{m\geq C_1}$ or $\{|ma+b|\}_{m\geq C_2}$ is a desired increasing unbounded sequence, where $C_1,C_2\in\mathbb{N}$ are constants.

\textit{Necessity.} In the infinite-dimensional algebra, there always exists a pair satisfying the conditions (i)-(v). 
\end{proof} 

\section{Implementation}
\label{sec:Implement}

At the core of our algorithms is the symbolic manipulation with Newton polytopes. We use the Maple-based Convex library \cite{franz2009convex} for these purposes.
It was developed by Franz based on the theory of Kushnirenko's bound for the sum of the Milnor numbers of the complex polynomial singularities \cite{kouchnirenko1976polyedres}. To our knowledge, it is the most efficient library for symbolic tasks in convex geometry.

\begin{algorithm}{\textsl{\bfseries{Check Superposition: Dimension 1}}\,($operators$)
\label{check-superpos1}}
\begin{algorithmic}[1]
\INPUT $operators$, list of polynomial vector fields
\OUTPUT {\tt True}, if it is closed under Lie bracket,\\\,\,\,\,\,\,\,\,\,\,\,\,\,\,$\tt False$, otherwise
\STATE $dim:=${\textsl{\bfseries{DimLinSpace}}}\,$(${\textsl{\bfseries{LieAlgebra}}}\,($operators$));
\STATE $polys:=${\textsl{\bfseries{ExtractPolynomials}}}\,($operators$);  
\WHILE{True}  
\STATE $degrees:=convert(map(degree,polys),set)$; \STATE $m_1:=max(degrees)$; 
\STATE $m_2:=max(degrees\setminus\{max(degrees)\})$;   
\IF {$(m_1>1)\wedge(m_2>1)$} 
\RETURN {\tt False};
\ELSE   
\STATE $operators:=${\textsl{\bfseries{AddPairwiseCommutators}}}\,($operators$); 
\IF{{\textsl{\bfseries{DimLinSpace}}}\,$(${\textsl{\bfseries{LieAlgebra}}}\,($operators$)) $>dim$}
\STATE $dim:=dim(${\textsl{\bfseries{DimLinSpace}}}\,$(${\textsl{\bfseries{LieAlgebra}}}\,($operators$)); 
\ELSE
\RETURN {\tt True};
\ENDIF
\ENDIF
\ENDWHILE
\end{algorithmic}
\end{algorithm}

\begin{algorithm}{\textsl{\bfseries{Check Superposition: General Case}}\,($operators$)
\label{check-superpos}}
\begin{algorithmic}[1]
\INPUT $operators$, list of polynomial vector fields
\OUTPUT {\tt True}, if it is closed under Lie bracket,\\\,\,\,\,\,\,\,\,\,\,\,\,\,\,$\tt False$, otherwise
\STATE $dim:=${\textsl{\bfseries{DimLinSpace}}}\,$(${\textsl{\bfseries{LieAlgebra}}}\,($operators$));
\STATE $polys:=${\textsl{\bfseries{ExtractPolynomials}}}\,($operators$);
\STATE $polytopes:=${\textsl{\bfseries{NewtonPolytopes}}}\,($operators$);
\STATE $cycle:=true$  
\WHILE{cycle}  
\FOR{$k=1$ to $numelems(polytopes)$} 
\FOR{$l=1$ to $numelems(polytopes)$} 
\IF{$k>l$}
\FOR{$a\in vertices(polytopes[k])$}
\FOR{$b\in vertices(polytopes[k])$}
\STATE $c:=a+b$;
\STATE $v_a:=${\textsl{\bfseries{GetCoeffs}}}\,$(a)$;
\STATE $v_b:=${\textsl{\bfseries{GetCoeffs}}}\,$(b)$;
\STATE 
$v_c:=${\textsl{\bfseries{DotProduct}}}\,$(b, v_a)v_b-${\textsl{\bfseries{DotProduct}}}\,$(a, v_b)v_a$;
\STATE 
$s_1:=0$;
\STATE 
$s_2:=0$;
\IF{{\textsl{\bfseries{DotProduct}}}\,$(b, v_a)=0$} 
\IF{{\textsl{\bfseries{DotProduct}}}\,$(b, v_b)\neq0$} 
\STATE
$s_1:=${\textsl{\bfseries{DotProduct}}}\,$(a, v_b)/${\textsl{\bfseries{DotProduct}}}\,$(b, v_b)$; 
\ENDIF
\ENDIF
\IF{{\textsl{\bfseries{DotProduct}}}\,$(a, v_b)=0$} 
\IF{{\textsl{\bfseries{DotProduct}}}\,$(a, v_a)\neq0$} 
\STATE
$s_2:=${\textsl{\bfseries{DotProduct}}}\,$(b, v_a)/${\textsl{\bfseries{DotProduct}}}\,$(a, v_a)$; 
\ENDIF
\ENDIF
\IF{$|c|>max(|a|,|b|)\wedge$\\ 
$c\in vertices(${\textsl{\bfseries{MinSum}}}\,$(polytopes[k], polytopes[l])$\\
$\wedge|v_c|\neq0\wedge(type(s_1,integer)=false\vee 0\leq s_1) \wedge$\\
$(type(s_2,integer)=false\vee 0\leq s_2)$
}
\RETURN {\tt False};
\ENDIF
\ENDFOR
\ENDFOR
\ENDIF
\ENDFOR
\ENDFOR
\STATE $operators:=${\textsl{\bfseries{AddPairwiseCommutators}}}\,($operators$); 
\IF{{\textsl{\bfseries{DimLinSpace}}}\,$(${\textsl{\bfseries{LieAlgebra}}}\,($operators$)) $>dim$}   
\STATE $dim:=dim(${\textsl{\bfseries{DimLinSpace}}}\,$(${\textsl{\bfseries{LieAlgebra}}}\,($operators$)); 
\ELSE
\RETURN {\tt True};
\ENDIF
\ENDWHILE
\end{algorithmic}
\end{algorithm}

The implementations of our two algorithms -- for the one-dimensional and for the general case -- are available on GitHub.\footnote{\url{https://github.com/treverona/nonlinear-superposition}}

\section{Non-trivial Example}
\label{sec:Nontrivial}

In order to construct an example, we use a simple trick: The Symmetry algebra of an ODE is finite-dimensional if differential order $n \geq 2$ and it is obviously closed under the Lie bracket. We consider the following system with variables $u(w)$ and $v(w)$:
$$
\frac{du}{dw} = v, \frac{d^2 v}{dw^2} = 0\,.
$$
The Lie symmetry algebra is defined by infinitesimal generators:
$$
X := \eta_u(w,u,v)) \frac{\partial}{\partial u} + \eta_v(w,u,v)) \frac{\partial}{\partial v} + \xi_w(w,u,v)) \frac{\partial}{\partial w}\,.
$$
As linear space it has the dimension $10$ and consists of the following operators:
\begin{flalign*}
&\eta_u(w,u,v)) = \frac{1}{2} \left( C_2 v^2 + C_1 \right) w^2 + \frac{1}{2} \left( C_4 v^2 + 2 C_3 u + 2 C_6 \right) w - \\
&- 2 C_2 u^2 + \frac{1}{2} \left( 2 C_8 + 2 C_5 \right) u + \frac{1}{2} C_7 v^2 + C_{10}\,, \\
&\eta_v(w,u,v)) = C_2 v^2 w + C_1 w - 2 C_2 u v + C_3 u + \frac{1}{2} C_4 v^2 + C_5 v + C_6\,, \\
&\xi_w(w,u,v)) = \frac{1}{2}(2 C_2 v + C_3) w^2 + (-2 C_2 u + C_4 v + C_8) w - C_4 u + C_7 v + C_9\,.
\end{flalign*}
The extraction of the two-dimensional Lie subalgebra leads to
\begin{flalign*}\label{nontrivial-example}
\frac{d}{d t}u(t) &=
\mathit{C_1}(t) \left(\frac{v(t)^{2} w(t)^{2}}{2}-2 u(t)^{2}\right)+\mathit{C_2}(t) u(t) w(t)\,,\\
\frac{d}{d t}v(t) &=
\mathit{C_1}(t) \left(\frac{v(t)^{2} w(t)}{2}-2 u(t) v(t)\right)+\mathit{C_2}(t) u(t)\,,\\
\frac{d}{d t}w(t) &=
\mathit{C_1}(t) \left(\frac{v(t) w(t)^{2}}{2}-2 u(t) w(t)\right)+\mathit{C_2}(t) \frac{w(t)^2}{2}\,.
\end{flalign*}

\section{Difference Thomas Decomposition}
\label{sec:Thomas}

The construction of numerical schemes is a challenging topic in mathematical modeling of natural phenomena \cite{thomas2013numerical}. From a naive point of view, all Cauchy problems for first-order ODE systems are covered by a basic theorem of Leonhard Euler, which is simple to use because of its explicit character.

\begin{theorem}(\cite{leveque2007finite}, pp.~6.3.3)\label{euler-theorem}
Let us suppose a first-order differential system
$$
\frac{dx}{dt} = f(x,t)
$$
with initial value $x(t_0) = x_0$, and smooth right-hand side $f$. The unique solution $x(t),\,t \in (a,b)$ defined by the Cauchy problem could be approximated by an explicit scheme
$$
\tilde{x}_{n+1} - \tilde{x}_{n} = h f(x_n, t_n)\,,
$$
and it converges for $x_n \rightarrow x(t)$ in $||\cdot||_{\infty}$ if $|x(t)| < \infty, \,t \in~(a,b)$.
\end{theorem}
In practice, it is required to apply very small time steps using explicit methods. Including the fact that numerical computations are carried out with floating-point arithmetic -- after a lot of steps -- a significant accumulation of errors occurs. That is why the explicit Euler method becomes unpractical, and to handle robust computation, it is important to care about the geometric property of the solution, such as symplecticity, and physical properties, such as conservation laws. Informally speaking, the general rule is to preserve as many properties as possible at a numerical level.

On the other hand, linearity is the most important property of differential equations, and for linear equations, without even thinking about it, we usually construct a linear finite-difference scheme. Remarkably, nonlinear superposition plays almost the same role for differential equations, which makes it the first property that we would like to preserve when we solve the problem numerically (as clearly seen in Section~\ref{sec:Riccati}).

The construction of numerical schemes can benefit greatly from the analog of differential the Thomas decomposition \cite{bachler2012algorithmic} adopted for difference equations. Potentially, for a finite-difference scheme with undefined parameters and known superposition rule (like the subharmonic ratio in \ref{subharmonic-ratio}), it will yield a desirable scheme or show that it is not possible to obtain it. There is ongoing research in the field based on the algebraic Thomas decomposition and specific properties of differences ideals \cite{gerdt2019algorithmic, gerdt2020strong}. It gives hope that progress in the field will lead to an opportunity to construct numerical schemes with advanced properties.

\section{Matrix Ricatti Equation}
\label{sec:MatrixRicatti}

Scalar Ricatti equations admit wonderful generalizations to matrix ones\footnote{Certain authors employ the phrase ``multidimensional Riccati equation'' to encompass all quadratic equations, whereas we prefer to reserve the term ``Riccati'' specifically for those equations that retain the nonlinear superposition property.}, which is worth to mention in our manuscript \cite{penskoi2004discrete}:
\begin{equation}\label{matrix-ricatti}
\frac{dW}{dt} = A(t) + B(t) W + W C(t) + W D(t) W\,,
\end{equation}
where $W$ is an $(n \times k)$-matrix with arbitrary $n, k \in \mathbb{N}$ and $A, B, C, D$ are matrices of appropriate dimensions.

The beautiful result that a two-parameter class of operators $U_{q,h}$ with

\begin{equation*}
U_{q,h} f(t) =
\begin{cases}
f'(t)\,, & \text{if } q = 1\,, h = 0\,, \\
\frac{f(qt+h) - f(t)}{(q-1)t + h}\,, & \text{else}\,. 
\end{cases}
\end{equation*}
defines the difference analog of the Ricatti equation, which admits the superposition formula
$$
U_{q,h} w(t) = a(t) + b(t) w(t) + w(qt+h) c(t) + w(qt+h) d(t) w(t)\,,
$$
and in particular for the case $q = 1$ and the limit case $h \rightarrow 0$, it converges to the differential matrix Ricatti equation. We will show that in addition to nonlinear superposition, it is possible to construct a scheme that is also strongly consistent \cite{blinkov2018strongly,michels2019consistency}.

Let us briefly remind the definition of strong consistency. We consider polynomially nonlinear ordinary (partial) differential equation systems of the form
\begin{equation}
f_1=\dots=f_p=0,\quad F:=\{f_1,\ldots,f_p\}\subset {\mathcal{R}}\,,
\label{pde}
\end{equation}
where $f_1,\ldots, f_p$ are elements of the differential polynomial ring $${\cR}:={\K}[t^{1},\ldots,t^{m}]$$ over a differential coefficient field ${\K}$. The differential polynomial ring ${\cR}$ contains polynomials in the dependent variables 
$$\mathbf{w}:=\{w^{1},\ldots,w^{n}\}$$ and their partial derivatives $\partial_{x_1}, \ldots,\partial_{x_n}$.

The coefficients on the grid are elements of the difference field with mutually commuting difference operators $\{\sigma_1,\ldots,\sigma_n\}$ acting on a function $\phi(\mathbf{x})$ as follows:
\begin{equation*}
\sigma_i\circ
\phi(x_1,\ldots,x_n)=\phi(x_1,\ldots,x_i+h_i,\ldots,x_n)\,,\quad \,h_i>0\,.
\label{rs-operators}
\end{equation*}
We refer to $\sigma_i$ also as shift operators. Analogously, we define $\tilde{\cR}$ as a difference polynomial ring.

\begin{definition}\label{def:10} We shall say that a difference equation $\tilde{f}(\mathbf{u})=0$ defined on the orthogonal and uniform grid with the grid spacing set $\mathbf{h}:=(h_1,\ldots,h_n)$ {\em implies the differential equation} $f(\mathbf{u})=0$ and write
$\tilde{f}\rhd f$ if the Taylor expansion about a grid point yields
\[
\tilde{f}(\mathbf{u})\xrightarrow[h_i\rightarrow 0]{} f(\mathbf{u}) + O(\mathbf{h})\,,
\]
where $O(\mathbf{h})$ denotes terms that reduce to zero when $h_i\rightarrow 0$ for $i=1,\ldots,n$.
\end{definition}

In the theory of strongly consistence finite-difference approximations, it is important to consider perfect difference ideals.

\begin{definition}{\em\cite{levin2008difference}} A {\em perfect difference ideal} generated by a set $\tilde{F}\in \tilde{\mathcal{R}}$ and denoted by $\llbracket \tilde{F}\rrbracket$ is the smallest difference ideal containing $\tilde{F}$ and such that for any $\tilde{f}\in {\mathcal{R}}$, $\theta_1,\ldots,\theta_r\in \Theta$ and $k_1,\ldots,k_r \in \N_{\geq 0}$:
\[
(\theta_1\circ \tilde{f})^{k_1}\dots (\theta_r\circ \tilde{f})^{k_r}\in \llbracket \tilde{F}\rrbracket \Longrightarrow \tilde{f}\in \llbracket \tilde{F}\rrbracket \,.\]\label{perfect}
\end{definition}
\vskip -0.6cm
In the context of difference algebra, perfect ideals play the same role as radical ideals in differential algebra. The clarity of this fact becomes evident when examining the difference analogue of the Nullstellensatz~\cite{trushin2009difference}.  For this reason, we will contemplate the perfect ideal $\llbracket \tilde{F}\rrbracket$ generated by the difference polynomials in the finite-difference approximation as the set of its difference-algebraic consequences. Respectively, the set of differential-algebraic consequences of a differential system is the radical differential ideal generated by the set $F$ in (\ref{pde}).
\begin{definition}{\em\cite{gerdt2019algorithmic}}
A finite-difference approximation to a differential system (\ref{pde}) is {\em strongly consistent} if
\begin{equation*}
(\,\forall \tilde{f}\in \llbracket \tilde{F} \rrbracket\,)\  (\,\exists
f\in \llbracket F \rrbracket \,)\ [\,\tilde{f}\rhd f\,]\,. \label{s-cond}
\end{equation*}
\label{def-scon}
\end{definition}
\vskip -0.6cm
Now, let us consider the square matrix case ($k = n$). Based on the operator $U_{q,h}$, we can construct a finite-difference scheme with desired properties.
\begin{theorem}
The finite difference numerical scheme 
\begin{equation}\label{difference-matrix-system}
\frac{W(t_{n+1}) - W(t_n)}{h} = A(t_n) + B(t_n) W(t_n) + W(t_{n+1}) (C(t_n) + D(t_n) W(t_n))\,,
\end{equation}
where $t_{n+1} = t_{n} + h$ with time step $h$, satisfies the three properties:

\begin{itemize}
\item $W(t_{n}) \rightarrow W(t_{n+1})$ is an invertible linear mapping\,;
\item difference solution admits nonlinear superposition formula\,;
\item as a numerical difference scheme, it is strongly consistent\,.
\end{itemize}
\end{theorem}

\begin{proof}
The explicit solution is given by formula
$$
W(t_{n+1}) = (I - hC(t_n) - hD(t_n)W(t_n))^{-1} (W(t_n)(I + hB(t_n)) + hA(t_n))\,,
$$
where the invertible matrix always exists for moderately small $h > 0$. For more details, we refer to the original paper \cite{penskoi2004discrete}.
Let us suppose it is not strongly consistent. Then, there exists an algebraic consequence of the difference matrix system (\ref{difference-matrix-system}) which in the limit of $h \rightarrow 0$ leads to the differential polynomial $S$ that is not lying in the radical of the differential ideal generated by the differential system of the matrix Ricatti equation (\ref{matrix-ricatti}). As it is not in the radical, we can find the solution $\tilde{W}$ of the differential system which does not satisfy $S$. System (\ref{matrix-ricatti}) is a well-posed Cauchy problem, so that every solution could be encoded by its initial value at some point (e.g., $\tilde{W}(t_0) = W_1$). In the difference setting, we can also find a solution with initial value $W_1$ at point $t_0$ which will converge to $\tilde{W}$ when $h \rightarrow 0$ by an analog of Theorem \ref{euler-theorem} (generalization to implicit and semi-explicit schemes) which immediately implies a contradiction as it should satisfies $S$.
\end{proof}

\section{Generalization and Open Questions}\label{sec:Generalization}
Our approach raises questions for several natural generalizations. Firstly, it is interesting to consider monomials on the right-hand side, which allow to have negative integer exponents. If functions $\xi^{i}_{k}(x)$ are represented as power series in the form of computable coefficients, is it still possible to decide algorithmically when the set is closed under the Lie bracket?

As we observe, the numerical scheme (\ref{good-scheme}) even outperforms the one obtained from Theorem \ref{euler-theorem} because it can cross the singularity points without an excessive accumulation of errors. It shows the importance of studying numerical schemes which preserve the property of nonlinear superposition.

\textbf{Open problem:} There are several ways how to transform a polynomial ODE system into a quadratic one by introducing new variables \cite{bychkov2021optimal,kerner1981universal}. Is it possible to transform a polynomial system into a Ricatti matrix system defined by (\ref{matrix-ricatti})?

\section{Conclusion}
\label{sec:Conclusion}

This contribution is devoted to the application of the theory of Newton polytopes for the verification of nonlinear superposition of first-order differential systems. Our viewpoint is that the presented methods assisted by computer algebra will be helpful for the construction and evaluation of appropriate discretizations of differential equations that appear in the applied sciences.

Also, we would like to emphasize that a lot of effort was invested into the field of Lie Algebras to classify all of them for some fixed dimension \cite{popovych2003realizations}. Typically, such kind of work involves the compilation of tables that describe all possible cases up to isomorphisms \cite{ibragimov2016three, ibragimov2017classification}. We truly believe that the development of symbolic algorithms provides the opportunity to extend powerful Lie formalism to a bigger audience.

Significant recent advancements in geometry \cite{trinh2024solving} demonstrate the effectiveness of machine learning as a tool for autonomously solving olympiad geometry problems. This underscores the significance of advancing symbolic analysis tools, which can serve as foundational components for future machine learning-driven automatic reasoning in algebra.

We conclude by anticipating that our approach will facilitate meaningful discussions and inspire productive future research endeavors.

\section*{Acknowledgments}
We wish to honor the memory of Vladimir P.~Gerdt, a real enthusiast for employing computer algebra in solving differential equations, with a particular emphasis on preserving their algebraic properties at a numerical level. This work has been partially supported by the baseline funding of the KAUST Computational Sciences Group. Veronika Treumova acknowledges KAUST support and hospitality during her research stay at KAUST.

\bibliographystyle{unsrt}
\bibliography{main}

\appendix
\section{Nonlinear superposition function}
The problem that was not discussed in the main text of the manuscript is how to find the superposition rule. This task could not be solved algorithmically as it requires solving a system of partial differential equations. However, it is still an interesting task from a symbolic point of view and requires proper heuristics. We illustrate it with an example of the scalar Ricatti equation. We need to prolong symmetry operators by adding their copies 
\begin{flalign*}
&eq1 := \frac{\partial F}{\partial x} + \frac{\partial F}{\partial x_1} + \frac{\partial F}{\partial x_2} + \frac{\partial F}{\partial x_3} = 0\,, \\
&eq2 := x \frac{\partial F}{\partial x} + x_1 \frac{\partial F}{\partial x_1} + x_2 \frac{\partial F}{\partial x_2} + x_3 \frac{\partial F}{\partial x_3} = 0\,, \\
&eq3 := x^2 \frac{\partial F}{\partial x} + x_1^2 \frac{\partial F}{\partial x_1} + x_2^2 \frac{\partial F}{\partial x_2} + x_3^2 \frac{\partial F}{\partial x_3} = 0\,.
\end{flalign*}
The built-in Maple command \textit{pdsolve} cannot find any solution. However, if we integrate only the first two equations, then it yields the solution
$$
F(x, x_1, x_2, x_3) = F_1 \left(\frac{-x + x_2}{x - x_1}, \frac{-x + x_3}{x - x_1} \right)\,.
$$
The change of independent variables
$$
(x,x_1,x_2,x_3) \mapsto (w,v,x_2,x_3)\,,
$$
where
$$
w = \frac{-x + x_2}{x - x_1}, \,v =  \frac{-x + x_3}{x - x_1}
$$
defines a local diffeomorphism almost everywhere because the Jacobi matrix
$$
\left[\begin{array}{cccc}
\frac{\mathit{x_1} -\mathit{x_2}}{\left(x -\mathit{x_1} \right)^{2}} & \frac{-x +\mathit{x_2}}{\left(x -\mathit{x_1} \right)^{2}} & \frac{1}{x -\mathit{x_1}} & 0 
\\
 \frac{\mathit{x_1} -\mathit{x_3}}{\left(x -\mathit{x_1} \right)^{2}} & \frac{-x +\mathit{x_3}}{\left(x -\mathit{x_1} \right)^{2}} & 0 & \frac{1}{x -\mathit{x_1}} 
\\
 0 & 0 & 1 & 0 
\\
 0 & 0 & 0 & 1 
\end{array}\right]
$$
has a non-zero determinant (except for the set of zero measure),
$$
\frac{\mathit{x_2} -\mathit{x_3}}{\left(x -\mathit{x_1} \right)^{3}} \ne 0\,.
$$
In new variables, the remaining equation $eq3$ could be rewritten as
$$
\left(w +1\right) w \left(\frac{\partial}{\partial w}\mathit{F_1} \! \left(w , v\right)\right)+\left(v +1\right) v \left(\frac{\partial}{\partial v}\mathit{F_1} \! \left(w , v\right)\right) = 0\,,
$$
which could be easily be integrated as
$$
\mathit{F_1} \! \left(w , v\right) = 
\mathit{F_2} \! \left(-\frac{v -w}{v \left(w +1\right)}\right)\,.
$$
Substituting $w$ and $v$ via $x_1, x_2, x_3, x$ and assuming $x = x_4$, we immediately obtain (\ref{subharmonic-ratio}). The general problem of constructing a superposition rule is out of scope for this paper, but we find it interesting as an applied problem in symbolic computation. Generally speaking, the nonlinear superposition should not distinguish which is the first and which is the second solution. Hence, it is required to construct a solution that is also a symmetric function of its variables.


\end{document}